\documentclass[journal]{IEEEtran}
\usepackage{cite}

\ifCLASSINFOpdf
   \usepackage[pdftex]{graphicx}
\else
\fi
\usepackage[cmex10]{amsmath}
\usepackage{amssymb,amsthm}
\usepackage[tight,footnotesize]{subfigure}

\usepackage[caption=false,font=footnotesize]{subfig}
\providecommand{\abs}[1]{\left|#1\right|}
\newtheorem{thm}{Theorem}
\begin{document}
%
\title{Incorporating Views on Market Dynamics\\in Options Hedging}
%
%
%

\author{Antoine~E.~Zambelli
\thanks{This is a preprint. A preliminary version of this work was presented at SIAM FM14 and IECSMA'15.}}

%
%

\markboth{Journal of \LaTeX\ Class Files,~Vol.~6, No.~1, January~2007}%
{Shell \MakeLowercase{\textit{et al.}}: Bare Demo of IEEEtran.cls for Journals}
%



\maketitle
\thispagestyle{empty}

\begin{abstract}
We examine the possibility of incorporating information or views of market movements during the holding period of a portfolio, in the hedging of European options with respect to the underlying. Given a fixed holding period interval, we explore whether it is possible to adjust the number of shares needed to effectively hedge our position to account for views on market dynamics from present until the end of our interval, to account for the time-dependence of the options' sensitivity to the underlying. We derive an analytical expression for the number of shares needed by adjusting the standard Black-Scholes-Merton $\Delta$ quantity, in the case of an arbitrary process for implied volatility, and we present numerical results.
\end{abstract}



%
\IEEEpeerreviewmaketitle

\section{Introduction}
%
%
%
%
\IEEEPARstart{W}{e} seek to incorporate any views, or information (referred to interchangeably) into the hedging of options. Namely, views on the growth rate of the stock and views on the implied volatility during the holding period of our portfolio. Much work has been done on optimal re-hedging frequency to improve the efficiency of hedges in real-world market conditions (see~\cite{zaka}). However, we base much of our work on~\cite{mastinsek}, which examines the effects of adjusting $\Delta$-hedges with Charm for short-term maturity options. We will adapt and extend this approach from the viewpoint of incorporating information.

Suppose we have a market-neutral options strategy. We enter into a Long position in a European Call option at $t_0$. Here, we introduce a new piece to the puzzle. Suppose we know from experience (or constraints) that we will hold our position until at least a later time $t_1 = t_0 + \Delta t$, whereupon we can exit our position at any time. Classically, we would hedge at $t_0$, re-balancing as needed over the interval $[t_0,t_1]$.

Without re-balancing, due to the market evolving over time, our hedge would no longer be perfect at $t_1$. The question we seek to answer is therefore: Given that we know our holding period interval, can we preemptively adjust our hedge at $t_0$ such that we are perfectly hedged at $t_1$? Or at least better-hedged? What if we have views on the market dynamics over the interval $[t_0,t_1]$?

We assume that our views of other market dynamics exists only on $[t_0,t_1]$. From a practitioner's standpoint, this assumption is akin to the notion that if we do have a lower bound $t_1$ for our holding period, then we can presumably run accurate conditional statistics on market dynamics on past intervals of that length.

\section{Preliminary Look}\label{sec:1}

Given a $\Delta$-hedged portfolio consisting of a Long position in a European call option, hedged at $t_0$
\begin{equation}
	\Pi(t)=V(t)-N(t_0)S(t)=V(t)-V_S(t_0)S(t)
\end{equation}
we wish to have
\begin{equation}
	\Pi_S(t_1)=V_S(t_1)-V_S(t_0)=0
\end{equation}
In the standard Black-Scholes-Merton framework from~\cite{black}, the number of shares $N(t_0)$ we need to short is simply $V_S(t_0)$. Unfortunately, this is unlikely to be accurate, as $V_S$ will likely change with time-varying market conditions. Since we know that $V_S$ is time-varying, we can attempt to correct for its time-decay component by setting
\begin{equation}
	N(t_0)=V_S(t_0)+V_{St}(t_0)\Delta t
\end{equation}
However, this assumes that $V_S$ changes linearly with time. We can include a nonlinear correction parameter $\lambda$, and define
\begin{equation}
	N(t_0)=V_S(t_0)+\lambda V_{St}(t_0)\Delta t
\end{equation}
Though we must still determine $\lambda$.

\subsection{The Literature}

While $\lambda$ can be approximated numerically,~\cite{mastinsek} determines it analytically (to $O(\Delta t^2)$). He treats this from the point of view of first-order hedging effects for short-term maturity options, but this is the same problem (mathematically) as we are attempting to solve. The author defines the hedging error
\begin{equation}\label{eq:DelH}
	\Delta H=\Delta \Pi -\Pi r\Delta t
\end{equation}
and attempts to minimize it. Namely, he arrives at a formula involving the growth rate $\mu$ of the underlying, and concludes that if $\mu=r$, we get back to a BSM hedge, but if $\mu\neq r$, the Mean Squared Hedging Error (MSHE) can be minimized by adjusting the number of shares $N$ by defining
\begin{equation}\label{eq:lambdaM}
	\lambda_M=\frac{(\mu-r)SV_{SS}}{V_{St}} 
\end{equation}
and
\begin{equation}\label{eq:NM}
	N_M=V_S+(\mu-r)SV_{SS}\Delta t
\end{equation}
However, as seen in~\cite{primbs}, minimizing MSHE may not be the best option, opting instead work mainly with the Mean Absolute Hedging Error (MAHE).

\subsection{Extension to Implied Volatility}

Looking at the literature from the point of view of incorporating views, we now have a plausible solution for incorporating information about two quantities: the growth rate of the underlying $\mu$, and the time at which we may liquidate $t_1$ (ie, $\Delta t$).

The natural extension of this work would be to account for any views on implied volatility $\Delta \sigma$. Note that we are not making any changes to what we believe the volatility of the underlying $\hat{\sigma}$ will do, we keep the standard BSM assumption that it will be constant. To summarize, the question is now the following: For a portfolio of the form
\begin{equation}
	\Pi(t)=V(t)-N(t_0) S(t)
\end{equation}
can we adjust the number of shares required to hedge (with respect to the underlying) our Long call position to account for both the $V_{St}$ and $V_{S\sigma}$ components? We work under the following assumptions: the underlying follows a standard Geometric Brownian Motion process
\begin{equation}\label{eq:gbm}
	S(t)=S_0e^{\left(\mu-\frac{1}{2}\hat{\sigma}^2\right)t+\hat{\sigma}W_1}
\end{equation}
and we consider an arbitrary stochastic process for implied volatility of the form
\begin{equation}\label{eq:IVstoch}
	d\sigma_t=f(\sigma_t)dt+g(\sigma_t)dW_2
\end{equation}
where $f,g$ are well-defined deterministic functions. For clarity of notation, we write $f_t,g_t$ for $f(\sigma_t),g(\sigma_t)$.

We also assume $V$ satisfies the $BSM$ equation
\begin{equation}\label{eq:bsm}
	V_t+\frac{1}{2}\hat{\sigma}^2S^2V_{SS}+rSV_S-rV=0
\end{equation}

Note that this is technically incorrect, as we will soon define a process for implied volatility, and the theory does not account for any difference in implied volatility and volatility of the underlying. In fact, there are no two quantities in the BSM framework, only one volatility. Nonetheless, this assumption allows us to simplify our equations a noticeable amount, and we will leave the case where we do not use equation~\eqref{eq:bsm} for a later version (much like using the $BSM$ option price formula after accounting for volatility smiles).

The natural extension of equation~\eqref{eq:NM} would be to set
\begin{equation}\label{eq:falseN}
	N=V_S+\lambda_1 V_{St}\Delta t+\lambda_2 V_{S\sigma}\Delta \sigma
\end{equation}
However, $\Delta\sigma$ is stochastic. This is where our assumption that the interval of information is fixed comes into play. If we have the endpoints of the interval, then in practice we could simulate the paths of the process and take the expected value, recasting the problem in a deterministic framework (explored in Section~\ref{sec:4}). More rigorously, let us take
\begin{equation}\label{eq:N}
	N=V_S+\lambda_1 V_{St}\Delta t+\lambda_2 V_{S\sigma}E\left[\Delta \sigma\right]
\end{equation}
We assume $W_1,W_2$ are uncorrelated. Note again that this formulation would not be as well defined given a random interval $[t_0,\tilde{t}]$.

\section{Derivation of the MSHE}\label{sec:3}

This section presents our main contribution, and derives an expression for the hedging error (to $O(\Delta t^2)$) in our framework and minimizes that error with respect to $\lambda_1,\lambda_2$ in equation~\eqref{eq:N}. This derivation is heavily inspired by~\cite{mastinsek}. A summary of the results, as well as numerical simulations is presented in Section~\ref{sec:4}.

\subsection{Derivation of $\Delta S$, $\Delta \sigma$ and $\Delta V$}\label{sec:31}

We begin with a series expansion of equation~\eqref{eq:gbm}, to obtain first, second, and third order terms:
\begin{equation}
	\Delta S=S_0e^{\left(\mu-\frac{1}{2}\hat{\sigma}^2\right)\Delta t+\hat{\sigma}\sqrt{\Delta t}Z}
\end{equation}
For clarity of notation, we omit $O(\Delta t^2)$ from the equations, though they are present throughout.
\begin{align}
	(\Delta S) =&{}\ S\biggl[ \hat{\sigma} Z_1\sqrt{\Delta t} +\left(\mu-\frac{1}{2}\hat{\sigma}^2\right)\Delta t+\frac{1}{2}\hat{\sigma}^2Z_1^2\Delta t\label{eq:DelS}\\
	&+\left(\frac{1}{6}\hat{\sigma}^3Z_1^2+\hat{\sigma}\left(\mu-\frac{1}{2}\hat{\sigma}^2\right)\right)Z_1\Delta t^{3/2} \biggr]\nonumber\\
	(\Delta S)^2=&{}\ S^2\biggl[ \hat{\sigma}^2Z_1^2\Delta t\label{eq:DelS2}\\
	&+2\hat{\sigma}\left( \left(\mu-\frac{1}{2}\hat{\sigma}^2\right)+\frac{1}{2}\hat{\sigma}^2Z_1^2 \right)Z_1\Delta t^{3/2} \biggr]\nonumber\\
	(\Delta S)^3&=S^3\hat{\sigma}^3Z_1^3\Delta t^{3/2}\label{eq:DelS3}
\end{align}
For $\Delta\sigma$, as per equation~\eqref{eq:IVstoch}, we have
\begin{align}
	(\Delta\sigma)&=f_0\Delta t +g_0\sqrt{\Delta t}Z_2\label{eq:delsig}\\
	(\Delta\sigma)^2&=2f_0g_0\Delta t^{3/2}Z_2+g_0^2\Delta t Z_2^2\label{eq:delsig2}\\
	(\Delta\sigma)^3&=g_0^3\Delta t^{3/2}Z_2^3\label{eq:delsig3} 
\end{align}
We now need to find an expression for $\Delta V$, which we find via a $3$-dimensional Taylor Expansion which gives us
\begin{align}\label{eq:DelV}
	\Delta V=&{} \ V_t\Delta t+V_S\Delta S+V_\sigma\Delta \sigma+V_{St}\Delta S\Delta t\\
	&+V_{S\sigma}\Delta S\Delta \sigma+V_{\sigma t}\Delta\sigma\Delta t +\frac{1}{2}V_{SS}(\Delta S)^2\nonumber\\
	&+\frac{1}{2}V_{\sigma\sigma}(\Delta\sigma)^2+\frac{1}{6}V_{SSS}(\Delta S)^3+\frac{1}{2}V_{SS\sigma}(\Delta S)^2(\Delta\sigma)\nonumber\\
	&+\frac{1}{2}V_{S\sigma\sigma}(\Delta S)(\Delta\sigma)^2+\frac{1}{6}V_{\sigma\sigma\sigma}(\Delta\sigma)^3\nonumber
\end{align}

\subsection{Derivation of $\Delta H$}\label{sec:32}

Now we derive an expression for equation~\eqref{eq:DelH}.
\begin{align}
	\Delta H&=\Delta \Pi - \Pi r \Delta t\\
	&=\Delta V - N\Delta S - \left(V-NS\right)r\Delta t\\
	&=\Delta V - Vr\Delta t - N\Delta S + rSN\Delta t\label{eq:DelHpreLem}
\end{align}
To simplify the expression, we use the property that $V$ satisfies equation~\eqref{eq:bsm}. Note also that $\Delta S\Delta\sigma=E\left[\Delta\sigma\right]\Delta S+g_0\sqrt{\Delta t}Z_2\Delta S$, and recall we are ignoring terms of $O(\Delta t^2)$ or greater. Now, for the $(\Delta S)^2$ term, take equation~\eqref{eq:DelS2}, and rewrite the $\hat{\sigma}^2Z^2\Delta t$ term as $\hat{\sigma}^2(Z^2-1+1)\Delta t$. Then separate out the resulting $S^2\hat{\sigma}^2\Delta t$ term. It follows therefore
\begin{align}\label{eq:DelHpreLem2}
	\Delta H=&{}\ (1-\lambda_1)V_{St}\Delta S\Delta t+(1-\lambda_2)V_{S\sigma}E\left[\Delta\sigma\right]\Delta S\\
	&+V_{S\sigma}g_0\sqrt{\Delta t}Z_2\Delta S+V_\sigma\Delta\sigma +V_{\sigma t}\Delta\sigma\Delta t\nonumber\\
	&+\frac{1}{2}V_{SS}\Gamma+\frac{1}{2}V_{\sigma\sigma}(\Delta\sigma)^2+\frac{1}{6}V_{SSS}(\Delta S)^3\nonumber\\
	&+\frac{1}{2}V_{SS\sigma}(\Delta S)^2(\Delta\sigma)+\frac{1}{2}V_{S\sigma\sigma}(\Delta S)(\Delta\sigma)^2\nonumber\\
	&+\frac{1}{6}V_{\sigma\sigma\sigma}(\Delta\sigma)^3\nonumber
\end{align}
where
\begin{align}
	\Gamma=&{}\ S^2\biggl[ \hat{\sigma}^2(Z_1^2-1)\Delta t\\
	&+2\hat{\sigma}\left( \left(\mu-\frac{1}{2}\hat{\sigma}^2\right)+\frac{1}{2}\hat{\sigma}^2Z_1^2 \right)Z_1\Delta t^{3/2} \biggr]\nonumber
\end{align}
We now use equation~\eqref{eq:bsm} a second time to further simplify the result by noting that
\begin{equation}
		V_{SSS}=\frac{-2}{\hat{\sigma}^2S^2}\left[\left(\hat{\sigma}^2S+rS\right)V_{SS}+V_{St}\right] 
\end{equation}
Applying this to equation~\eqref{eq:DelHpreLem2} and plugging in equations~\eqref{eq:DelS} to~\eqref{eq:delsig3} yields
\begin{align}\label{eq:DelHexplicit}
	\Delta H =&{}\ (1-\lambda_1)V_{St}\hat{\sigma}SZ_1\Delta t^{3/2}-\frac{1}{3}V_{St}\hat{\sigma}SZ_1^3\Delta t^{3/2}\\
	&+V_\sigma g_0\sqrt{\Delta t}Z_2+V_\sigma f_0\Delta t+\frac{1}{2}V_{SS}\Gamma\nonumber\\
	&-\frac{1}{3}(\hat{\sigma}^2S+rS)V_{SS}\hat{\sigma}SZ_1^3\Delta t^{3/2}+\frac{1}{2}V_{\sigma t}g_0\Delta t^{3/2}Z_2\nonumber\\
	&+(1-\lambda_2)V_{S\sigma}\hat{\sigma}SZ_1f_0\Delta t^{3/2}+V_{S\sigma}g_0S\hat{\sigma}\Delta t Z_1Z_2\nonumber\\
	&+V_{S\sigma}\frac{1}{2}S\hat{\sigma}^2g_0\Delta t^{3/2}Z_1^2Z_2\nonumber\\
	&+V_{S\sigma}Sg_0\left(\mu-\frac{1}{2}\hat{\sigma}^2\right)\Delta t^{3/2}Z_2\nonumber\\
	&+\frac{1}{2}V_{\sigma\sigma}g_0^2\Delta t Z_2^2+\frac{1}{2}V_{\sigma\sigma}2f_0g_0\Delta t^{3/2}Z_2\nonumber\\
	&+\frac{1}{2}V_{SS\sigma}S^2\hat{\sigma}^2Z_1^2g_0Z_2\Delta t^{3/2}\nonumber\\
	&+\frac{1}{2}V_{S\sigma\sigma}g_0^2Z_2^2S\hat{\sigma}Z_1\Delta t^{3/2}+\frac{1}{6}V_{\sigma\sigma\sigma}g_0^3\Delta t^{3/2}Z_2^3\nonumber
\end{align}
Now collect the terms in equation~\eqref{eq:DelHexplicit} to get an expression in terms of factors of $Z_1,Z_2$, giving us
\begin{align}\label{eq:DelHcoll}
	\Delta H=&{}\ \gamma\left[ \left(Z_1^2-1\right)+\beta Z_1 +\delta Z_1^3 \right]\\
	&+\phi\left[(1-\lambda_1)Z_1-\frac{1}{3}Z_1^3\right]+(1-\lambda_2)\eta Z_1+\omega Z_2\nonumber\\
	&+ \tau Z_1Z_2+\iota Z_2^2+\chi Z_1^2Z_2+\xi Z_1Z_2^2+\varepsilon Z_2^3 +V_\sigma f_0\Delta t \nonumber
\end{align}
where
\begin{equation}\label{eq:betadelta}
	\begin{aligned}[t]
		\gamma & = \frac{1}{2}V_{SS}\hat{\sigma}^2S^2\Delta t\quad \\
		\beta &= \frac{2\left(\mu-\frac{1}{2}\hat{\sigma}^2\right)}{\hat{\sigma}}\sqrt{\Delta t}\quad\\
		\delta &= \frac{\left(\hat{\sigma}^2-2r\right)}{3\hat{\sigma}}\sqrt{\Delta t}\\
		\phi & = V_{St}\hat{\sigma}S\Delta t^{3/2}\quad\\
		\eta &= V_{S\sigma}\hat{\sigma}Sf_0\Delta t^{3/2}\quad\\
		\varepsilon&=\frac{1}{6}V_{\sigma\sigma\sigma}g_0^3\Delta t^{3/2}\quad\\
		\xi &=\frac{1}{2}V_{S\sigma\sigma}g_0^2S\hat{\sigma}\Delta t^{3/2}\quad
	\end{aligned}
	\begin{aligned}[t]
		\omega=&{}\ V_\sigma g_0\sqrt{\Delta t}+\frac{1}{2}V_{\sigma\sigma}2f_0g_0\Delta t^{3/2}\\
		&+V_{S\sigma}Sg_0\left(\mu-\frac{1}{2}\hat{\sigma}^2\right)\Delta t^{3/2}\\
		&+g_0\Delta t^{3/2}\\
		\tau =& V_{S\sigma}g_0S\hat{\sigma}\Delta t\\
		\iota =&\frac{1}{2}V_{\sigma\sigma}g_0^2\Delta t\\
		\chi =&\frac{1}{2}V_{SS\sigma}S^2\hat{\sigma}^2g_0\Delta t^{3/2}\\
		&+\frac{1}{2}V_{S\sigma}S\hat{\sigma}^2g_0\Delta t^{3/2}\\
	\end{aligned}
\end{equation}
Simplifying further, we obtain
\begin{align}\label{eq:DelHcoll2}
	\Delta H =&{}\ \gamma\left(Z_1^2-1\right)+\theta Z_1 +\psi Z_1^3 +\omega Z_2 + \tau Z_1Z_2\\
	&+\iota Z_2^2+\chi Z_1^2Z_2+\xi Z_1Z_2^2+\varepsilon Z_2^3 +V_\sigma f_0\Delta t\nonumber
\end{align}
where
\begin{align}\label{eq:thetapsi}
	\theta & = \gamma\beta+(1-\lambda_1)\phi +(1-\lambda_2)\eta \\
	\psi &= \gamma\delta-\frac{1}{3}\phi
\end{align}

\subsection{Minimizing the MSHE}\label{sec:33}

Having obtained a suitable form for $\Delta H$, we now try to find $\lambda_1,\lambda_2$ that minimizes the MSHE. We first recall a property of the moments of $Z$. For $Z\sim N(0,1),\ n\in \mathbb{N}$, we have
\begin{equation}\label{rem:3}
	E\left(Z^{2n}\right)=\prod_{i=1}^{n}(2i-1)\ \textrm{and} \ E(Z^{2n-1})=0
\end{equation}

\begin{thm}\label{thm:1}
If $V$ satisfies the BSM equation, with
	\begin{align*}
		\Delta H(\lambda_1,\lambda_2)&=\Delta \Pi(\lambda_1,\lambda_2)-\Pi(\lambda_1,\lambda_2) r \Delta t\\ \Pi(\lambda_1,\lambda_2)&=V-N(\lambda_1,\lambda_2)S\\
		N(\lambda_1,\lambda_2)&=V_S+\lambda_1V_{St}\Delta t+\lambda_2V_{S\sigma}E[\Delta\sigma]
	\end{align*}
	then
	\begin{equation*}
		\exists\ \lambda_2^* \mid Var\left(\Delta H(\lambda_1,\lambda_2^*)\right)=\min_{\lambda_1,\lambda_2}\left\{Var\left(\Delta H(\lambda_1,\lambda_2)\right)\right\}
	\end{equation*}
	where
	\begin{equation}\label{eq:rhoc}
		\lambda_2^*=\frac{V_{S\sigma}f_0+\frac{1}{2}V_{S\sigma\sigma}g_0^2+V_{SS}(\mu-r)S-\lambda_1 V_{St}}{ V_{S\sigma}f_0}
	\end{equation}
	and
	\begin{equation}
		N^*=V_S+V_{SS}(\mu-r)S\Delta t+V_{S\sigma}f_0\Delta t+\frac{1}{2}V_{S\sigma\sigma}g_0^2\Delta t 
	\end{equation}
\end{thm}

\begin{proof}
	Applying equation~\eqref{rem:3} to equation~\eqref{eq:DelHcoll2}, we have
	\begin{align}
		Var\left(\Delta H\right)=&{}\ E\left[ \left(\Delta H \right)^2 \right]-\left(E\left[\Delta H\right]\right)^2\\
		=& \ E\Biggl[ \biggl( \gamma\left(Z_1^2-1\right)+\theta Z_1 +\psi Z_1^3 +\omega Z_2\label{eq:formmshe}\\
		&+ \tau Z_1Z_2+\iota Z_2^2+\chi Z_1^2Z_2+\xi Z_1Z_2^2\nonumber\\
		&+\varepsilon Z_2^3 +V_\sigma f_0\Delta t \biggr)^2 \Biggr]\nonumber\\
		&-\iota^2-2\iota\left( V_\sigma f_0\Delta t \right)-\left( V_\sigma f_0\Delta t \right)^2\nonumber
	\end{align}
	Sparing the reader the ensuing algebra, we obtain
	\begin{align}\label{eq:VarDelH}
		Var(\Delta H)=&{}\ \theta^2+15\psi^2+6\theta\psi+2\gamma^2+2\xi\theta +2\iota^2\\
		&+3\xi^2+6\xi\psi+3\chi^2+2\chi\omega+\omega^2+\tau^2\nonumber\\
		&+15\varepsilon^2+6\varepsilon\chi+6\varepsilon\omega\nonumber
	\end{align}
	Let $F=Var(\Delta H),\ a=1-\lambda_1,\ b=1-\lambda_2$. Expanding equation~\eqref{eq:VarDelH} with respect to $a,b$, we get
	\begin{align}
		F=&{}\ \phi^2a^2+\eta^2b^2+2\eta\phi ab+\left(2\gamma\beta\phi+2\xi\phi+6\psi\phi\right)a\label{eq:F}\\
		&+\left(2\gamma\beta\eta+2\xi\eta+6\psi\eta\right)b+\gamma^2\beta^2+2\xi\gamma\beta\nonumber\\
		&+6\psi\gamma\beta+15\psi^2+2\gamma^2+2\iota^2+3\xi^2+6\xi\psi+3\chi^2\nonumber\\
		&+2\chi\omega+\omega^2+\tau^2+15\varepsilon^2+6\varepsilon\chi+6\varepsilon\omega\nonumber
	\end{align}
	and
	\begin{align}
		F_a&=2\phi^2a+2\eta\phi b+2\gamma\beta\phi+2\xi\phi+6\psi\phi	\label{eq:Fa}\\
		F_b&=2\eta^2b+2\eta\phi a +2\gamma\beta\eta+2\xi\eta+6\psi\eta	\label{eq:Fb}
	\end{align}
	Note that equations~\eqref{eq:Fa} and~\eqref{eq:Fb}, have no dependence on the $-\iota^2-2\iota\left(V_\sigma f_0\Delta t\right)-\left(V_\sigma f_0\Delta t\right)^2$ term, which is not a coefficient of $a$ or $b$. Therefore, minimizing $Var(\Delta H)$ will also minimize the $MSHE$ (which itself contains that term).\\
	
	Furthermore, equation~\eqref{eq:F} has a linear dependence on $a,b$, so we can minimize with respect to those variables. Setting equation~\eqref{eq:Fb} to $0$, we get
	\begin{equation}
		b=\frac{-\gamma\beta-\phi a -3\psi-\xi}{\eta}
	\end{equation}
	Plugging in equations~\eqref{eq:betadelta}, we obtain
	\begin{equation}
		b=\frac{-V_{SS}(\mu-r)S\Delta t+(1-a)V_{St}\Delta t-\frac{1}{2}V_{S\sigma\sigma}g_0^2\Delta t}{ V_{S\sigma}f_0\Delta t}
	\end{equation}
	which gives us
	\begin{equation}\label{eq:rhoinproof}
		\lambda_2^*=\frac{V_{S\sigma}f_0+\frac{1}{2}V_{S\sigma\sigma}g_0^2+V_{SS}(\mu-r)S-\lambda_1 V_{St}}{ V_{S\sigma}f_0}
	\end{equation}
	Taking equation~\eqref{eq:rhoinproof} and plugging it into equation~\eqref{eq:N} gives us
	\begin{equation}
		N^*= V_S+V_{SS}(\mu-r)S\Delta t+V_{S\sigma}f_0\Delta t+\frac{1}{2}V_{S\sigma\sigma}g_0^2\Delta t 
	\end{equation}
\end{proof}

Our expression reduces to that of equation~\eqref{eq:NM} in the case where $f_0=g_0=0$, a comforting result. Additionally, in the cases where both $f_0=g_0=0$ and $\mu=r$, or where our interval $\Delta t\to 0$, our result reduces to the standard BSM hedge. However, we have no dependence on $\lambda_1$, as equations~\eqref{eq:Fa} and~\eqref{eq:Fb} are dependent. We can therefore simplify the framework a little more.

\subsection{Case: $\lambda_1=\lambda_2$}\label{sec:34}

Given the dependence exhibited above, set $\lambda=\lambda_1=\lambda_2$ in equation~\eqref{eq:N}. Note that the following can be derived independently, but working from the result of Theorem $3.2$ greatly reduces the work involved.

\begin{thm}
If $V$ satisfies the BSM equation, with
	\begin{align*}
		\Delta H(\lambda)&=\Delta \Pi(\lambda)-\Pi(\lambda) r \Delta t\\
		\Pi(\lambda)&=V-N(\lambda)S\\
		N(\lambda)&=V_S+\lambda\left(V_{St}\Delta t+V_{S\sigma}E[\Delta\sigma]\right)
	\end{align*}
	then
	\begin{equation*}
		\exists !\ \lambda^* \mid Var\left(\Delta H(\lambda^*)\right)=\min_{\lambda}\left\{Var\left(\Delta H(\lambda)\right)\right\}
	\end{equation*}
	where
	\begin{equation}\label{eq:lambdac}
		\lambda^*=\frac{V_{S\sigma}f_0+\frac{1}{2}V_{S\sigma\sigma}g_0^2+V_{SS}(\mu-r)S}{ V_{S\sigma}f_0+V_{St}}
	\end{equation}
	and
	\begin{equation}\label{eq:fullN}
		N^*=V_S+V_{SS}(\mu-r)S\Delta t+V_{S\sigma}f_0\Delta t+\frac{1}{2}V_{S\sigma\sigma}g_0^2\Delta t
	\end{equation}
\end{thm}
\begin{proof}
	From Theorem~\ref{thm:1}, we have
	\begin{equation}
		\lambda_2^*=\frac{V_{S\sigma}f_0+V_{SS}S(\mu-r)-\lambda_1 V_{St}+\frac{1}{2}V_{S\sigma\sigma}g_0^2}{V_{S\sigma}f_0}
	\end{equation}
	Let $\lambda^*=\lambda_1=\lambda_2^*$, then
	\begin{align}
		\lambda^*\left(1+\frac{V_{St}}{ V_{S\sigma}f_0}\right)&=\frac{ V_{S\sigma}f_0+\frac{1}{2}V_{S\sigma\sigma}g_0^2+V_{SS}(\mu-r)S}{V_{S\sigma}f_0}\\
		\lambda^*\left( \frac{ V_{S\sigma}f_0+V_{St}}{ V_{S\sigma}f_0}\right)&=\frac{V_{S\sigma}f_0+\frac{1}{2}V_{S\sigma\sigma}g_0^2+V_{SS}(\mu-r)S}{V_{S\sigma}f_0}
	\end{align}
	which gives us
	\begin{equation}\label{eq:lambdacinproof}
		\lambda^*=\frac{V_{S\sigma}f_0+\frac{1}{2}V_{S\sigma\sigma}g_0^2+V_{SS}(\mu-r)S}{V_{S\sigma}f_0+V_{St}}
	\end{equation}
	Plugging in equation~\eqref{eq:lambdacinproof} into our expression for $N$ gives us
	\begin{equation}
		N^*=V_S+V_{SS}(\mu-r)S\Delta t+ V_{S\sigma}f_0\Delta t+\frac{1}{2}V_{S\sigma\sigma}g_0^2\Delta t 
	\end{equation}
\end{proof}

While we have a closed-form solution to our problem, we did so with certain simplifying assumptions. Still, we can easily determine whether or not we have over-simplified by doing some numerical simulations.

\section{Summary and Results}\label{sec:4}

Using the method outlined above, we have shown that for a GBM process for the underlying security  and a stochastic implied volatility model as defined below (where $W_1,W_2$ are uncorrelated Weiner processes)
\begin{align}
	dS_t&=S_t\mu dt+S_t\hat{\sigma}dW_1\\
	d\sigma_t&=f_tdt+g_tdW_2
\end{align}
then the MSHE of the portfolio will be minimized by offsetting our options position with a number of shares given by
\begin{equation*}
	\boxed{N^*=V_S+V_{SS}(\mu-r)S\Delta t+ V_{S\sigma}f_0\Delta t+\frac{1}{2}V_{S\sigma\sigma}g_0^2\Delta t}
\end{equation*}

This gives us the answer to the question posed in the introduction of this paper. We can indeed preemptively adjust our hedge at $t_0$, such that we are better-hedged at $t_1$. As we can see, should the market remain static on this interval, then we conclude that the standard BSM hedge is the best as expected. Likewise for an instantaneous hedge (where $\Delta t\to 0$).

We present some results of our model for several possible implied volatility models. The stochastic models could be calibrated at $t_0$ in the usual ways, and those parameters used in the formula for $N^*$.

\subsection{Linear Drift}

Assuming a vanishing $g_t$ function and a constant growth rate $\mu_\sigma$, we have a deterministic model with only a drift term, given by
\begin{equation}
	d\sigma_t=\mu_\sigma dt
\end{equation}
While simple, this could be obtained by recasting the result of simulated paths' expected value (perhaps in the case of a more complex, intractable model). In this case, we obtain
\begin{equation}\label{eq:linearN}
	N^*=V_S+V_{SS}(\mu-r)S\Delta t+V_{S\sigma}\mu_\sigma\Delta t
\end{equation}

\subsection{Ornstein-Uhlenbeck}

The Ornstein-Uhlenbeck model (also known as Vasicek) allows for mean-reversion, and is given by

\begin{equation}
	d\sigma_t=\kappa\left(\bar{\theta}-\sigma_t\right)dt+\alpha dW
\end{equation}
where $\bar{\theta}$ is the long-run mean to which the process reverts, $\kappa$ is the rate at which it reverts to the mean, and $\alpha$ is the volatility of the process. In this case, we have
\begin{align}\label{eq:OUN}
	N^*=&{}\ V_S+V_{SS}(\mu-r)S\Delta t\\
	&+V_{S\sigma}\kappa\left(\bar{\theta}-\sigma_0\right)\Delta t+\frac{1}{2}V_{S\sigma\sigma}\alpha^2\Delta t \nonumber
\end{align}

\subsection{Cox-Ingersoll-Ross}

The CIR model, presented in the seminal work by~\cite{cir}, is given by
\begin{equation}
	d\sigma_t=\kappa\left(\bar{\theta}-\sigma_t\right)dt+\alpha\sqrt{\sigma_t} dW
\end{equation}
where the quantities are defined as in the Ornstein-Uhlenbeck model. The CIR model can guarantee positive volatility if the Feller condition is satisfied, that is: $2\kappa\bar{\theta}>\alpha^2$. Here, we obtain
\begin{align}\label{eq:cirN}
	N^*=&{}\ V_S+V_{SS}(\mu-r)S\Delta t\\
	&+V_{S\sigma}\kappa\left(\bar{\theta}-\sigma_0\right)\Delta t+\frac{1}{2}V_{S\sigma\sigma}\alpha^2\sigma_0\Delta t \nonumber
\end{align}

\subsection{Numerical Results}

For our numerical purposes, we take a simple linear drift model for implied volatility, with $f_0=\mu_\sigma,\ g_0=0$. Note that with this form, $N^*$ is given by equation~\eqref{eq:linearN}.

The following results were computed using $100000$ GBM paths on European Call options with fixed parameters: $S_0=K=100,\ T=0.1,\ r=0.05,\ \hat{\sigma}=\sigma_0=0.2,\ t_0=0,\ t_1=0.02$. We then vary our two parameters $\mu,\mu_\sigma$ on $[-0.5,0.5]$ and set $\sigma(t_1)=\sigma_0+\mu_\sigma\Delta t$. We then compute the MAHE with your adjusted hedge and compare it to the MAHE using the standard BSM hedge, ie, for
\begin{align}
	\Pi&=V-\left( V_S+V_{SS}(\mu-r)S+\left(\mu_\sigma V_{S\sigma}\right)\Delta t \right)S\\
	\Pi_{BSM}&=V-\left(V_S\right)S
\end{align}
we compute
\begin{align}
	E\left[\abs{\Delta H}\right]=&{} \ E\left[ \abs{ \Pi(t_1)-\Pi(t_0) -\Pi(t_0)r\Delta t } \right]\\
	E\left[\abs{\Delta H_{BSM}} \right]=&{} \ E\bigl[\bigl| \Pi_{BSM}(t_1)-\Pi_{BSM}(t_0)\\
	&-\Pi_{BSM}(t_0)r\Delta t\bigr|  \bigr]\nonumber
\end{align}
\begin{figure}[h!]
	\centering\includegraphics[height=4cm]{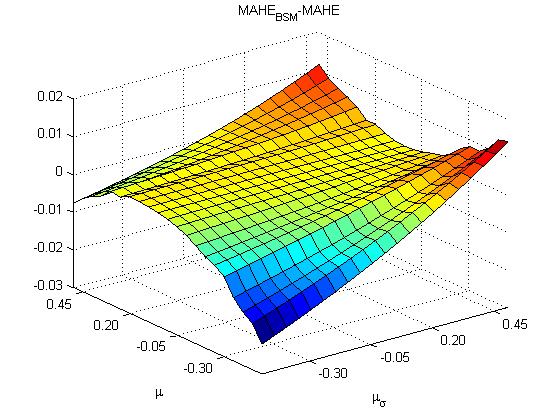}
	\caption{Difference between $E\left[\abs{\Delta H_{BSM}} \right]$ and $E\left[\abs{\Delta H}\right]$.}
	\label{fig:MAHEdiff}
\end{figure}
As we can see in Figure~\ref{fig:MAHEdiff}, there is a significant area for which we have improved the MAHE (while hard to see on the figure, it is also improved for $\mu=r$, see Table~\ref{tbl:1}). Though the apparent decline in performance of our method for $\mu_\sigma<0$ must be examined more closely, this is a promising first pass.
\begin{table}[h!]
\renewcommand{\arraystretch}{1.3}
	\footnotesize{\caption{Difference ($\textrm{to }O(10^{-4})$) in MAHE between BSM and extended models, for $\mu=r$.}
		\begin{center}
			\begin{tabular}{|c|*{12}{c|}}
				\hline
				$\mu_\sigma$ & -0.05& 0&  0.10&  0.20& 0.30& 0.40& 0.50\\
				\hline
				Difference&-0.06& 0& 0.08& 0.15& 0.17& 0.16& 0.12\\ \hline
			\end{tabular}\label{tbl:1}
		\end{center}
	}
\end{table}
To gain a better understanding of the shortcomings of this current formulation, we compare the performance using equation~\eqref{eq:linearN} to $\Delta H_M\equiv \Delta H(\lambda_M)$ (from equation~\eqref{eq:NM}).
\begin{figure}[h!]
	\centering\includegraphics[height=4cm]{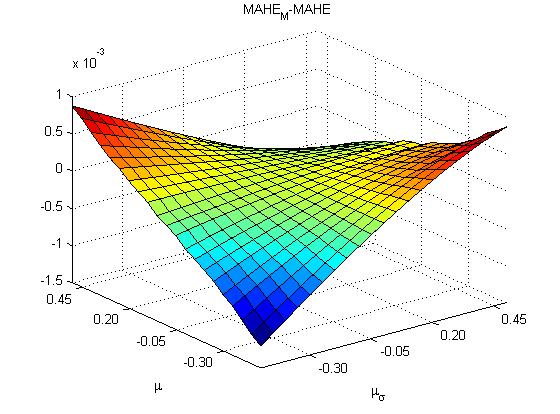}
	\caption{Difference between $E\left[\abs{\Delta H_M} \right]$ and $E\left[\abs{\Delta H}\right]$.}
	\label{fig:MAHEdiffM}
\end{figure}
As the parameters deviate from the BSM framework, we improve the MAHE in our portfolio considerably. Though apparent instability and sensitivity to our parameters should not be discounted.

\section{Conclusion}

We have successfully accounted for $3$ separate views, or pieces of information: the growth rate of the underlying, the change in implied volatility, and the duration of our holding period. This result would allow us to better hedge our position given information about what will (or may) occur during our holding period (or a well-defined subset thereof corresponding to $[t_0,t_1]$).

Our contributions in this work are therefore twofold. Firstly, re-framing the literature \cite{mastinsek} as a more general hedging problem. Secondly, extending this model for an arbitrary stochastic process for implied volatility.

However, we obtained this result with the use of the Black-Scholes-Merton equation, which does not technically apply in our situation. In addition, we have assumed that the stochastic components of the underlying security and implied volatility are independent.

Future work would include two main components. Firstly, deriving a new expression, similar to the one obtained thus far, but removing the reliance on the Black-Scholes equation. If our result still proves to exhibit sub-standard performance in certain areas, then a deeper analysis of those cases should be conducted. Finally, attempting to derive expressions in the case where our random processes $W_1,W_2$ are correlated.

\ifCLASSOPTIONcaptionsoff
  \newpage
\fi



\bibliographystyle{IEEEtran}
\end{document}